\documentclass[11pt]{article}

\usepackage{enumerate}
\usepackage{amsmath,amssymb,amsthm,hyperref,color}
\usepackage{amsbsy}
\usepackage{graphicx,subfigure}
\hypersetup{colorlinks=true,citecolor=blue, linkcolor=blue, urlcolor=blue}

\usepackage[margin = 1in]{geometry}
\usepackage{graphics}
\newtheorem{theorem}{Theorem}
\newtheorem{conjecture}{Conjecture}
\newtheorem{lemma}{Lemma}

\newtheorem{remark}{Remark}
\newtheorem{definition}[remark]{Definition}

\newcommand\Tree{\mathbb{T}}
\newcommand{\fptas}{\mathsf{FPTAS}}

\newcommand{\bD}{\boldsymbol{D}}

\newcommand{\bM}{\mathbf{M}}

\newcommand{\C}{\mathcal{C}}

\newcommand{\cF}{\mathcal{F}}

\newcommand{\beq}{\begin{eqnarray}}
\newcommand{\eeq}{\end{eqnarray}}
\newcommand{\beqn}{\begin{equation}}
\newcommand{\eeqn}{\end{equation}}

\newcommand{\Z}{\mathbb{Z}}

\newcommand{\mc}{\mathcal}

\renewcommand{\P}{\mc{P}}

\newcommand{\bB}{\mathbf{B}}

\newcommand{\dd}{\mathbf{d}}

\newcommand{\integers}{\mathbb{Z}}
\newcommand{\T}{\mathbb{T}}

\newcommand{\even}{\mathrm{even}}
\newcommand{\odd}{\mathrm{odd}}

\newcommand{\Tsaw}{T_{\mathrm{saw}}}
\newcommand{\dist}{\mathrm{dist}}
\newcommand{\rp}{p}

\newcommand{\newlbd}{{2.48}}

\newcommand{\trz}{\operatorname{tr}}

\title{Improved Bounds on the Phase Transition for the Hard-Core Model
in 2-Dimensions}

\author{
Juan C. Vera\thanks{Department of Econometrics and Operations Research,
Tilburg University,
5000 LE Tilburg,
The Netherlands.
Email: j.c.veralizcano@uvt.nl.
}
\and
Eric Vigoda\thanks{School of Computer Science, Georgia
Institute of Technology, Atlanta GA 30332.
Email: vigoda@cc.gatech.edu.
Research supported in part by NSF grant CCF-1217458.}
\and
Linji Yang\thanks{Facebook, Inc.  Menlo Park, CA.
Email: yang.linji@gmail.com.}
}

\begin{document}

\maketitle

\thispagestyle{empty}

\begin{abstract}
For the hard-core lattice gas model defined on independent sets weighted
by an activity~$\lambda$, we study the critical activity $\lambda_c(\integers^2)$
for the uniqueness/non-uniqueness threshold on the 2-dimensional integer lattice $\integers^2$.
The conjectured value of the critical activity is approximately $3.796$.
Until recently, the best lower bound followed from algorithmic results of Weitz (2006).
Weitz presented an $\fptas$ for approximating the partition function for graphs of constant
maximum degree $\Delta$ when $\lambda<\lambda_c(\Tree_\Delta)$ where
$\Tree_\Delta$ is the infinite, regular tree of degree $\Delta$.  His result
established a certain decay of correlations property
called strong spatial mixing (SSM) on $\integers^2$ by proving that
SSM holds on its self-avoiding walk tree $\Tsaw^\sigma(\integers^2)$ where
$\sigma=(\sigma_v)_{v\in \integers^2}$ and $\sigma_v$ is an ordering on the neighbors of vertex $v$.
As a consequence he obtained that $\lambda_c(\integers^2)\geq\lambda_c( \Tree_4) = 1.675$.
Restrepo et al.~(2011) improved Weitz's approach for the particular case of $\integers^2$
and obtained that $\lambda_c(\integers^2)>2.388$.  In this paper, we establish
an upper bound for this approach, by showing that, for all $\sigma$,
SSM does not hold on $\Tsaw^\sigma(\integers^2)$ when $\lambda>3.4$.
We also present a refinement of the approach of
Restrepo et al. which improves the lower bound to
$\lambda_c(\integers^2)>\newlbd$.
\end{abstract}

\newpage

\setcounter{page}{1}

\section{Introduction}
The hard-core model is a model of a gas composed of particles of non-negligible
size and consequently configurations of the model are independent sets~\cite{BergSteif,GauntFisher}.
For a (finite) graph $G=(V,E)$ and an activity $\lambda>0$ (corresponding to
the fugacity of the gas), configurations of the model are
the set $\Omega$ of independent sets of $G$ where
$\sigma\in\Omega$ has weight $w(\sigma) = \lambda^{|\sigma|}$.
The Gibbs measure is defined as $\mu(\sigma) = w(\sigma)/Z$ where
$Z=\sum_{\eta\in\Omega} w(\eta)$ is the partition function.

A fundamental question for statistical physics models, such as the
hard-core model, is whether there exists a unique or there are multiple
infinite-volume Gibbs measures on $\integers^2$. An equivalent question
is whether the influence of the boundary on the origin decays in the limit.
More formally, for a box in $\integers^2$ of side length $2L+1$ centered
around the origin, let $\rp_L^\even$ ($\rp_L^\odd$) denote the marginal
probability that the origin is unoccupied conditional on the even (odd, respectively)
vertices on the boundary being occupied.  If
\[  \lim_{L\rightarrow\infty} \left|\rp_L^\odd -\rp_L^\even\right| = 0
\] then there is a unique Gibbs measure on $\integers^2$, and
if this limit is $>0$ then there are multiple Gibbs measures.
It is believed that there is a critical activity $\lambda_c(\integers^2)$ such
that for $\lambda<\lambda_c(\integers^2)$ uniqueness holds,
and for $\lambda>\lambda_c(\integers^2)$ non-uniqueness holds.
For the infinite, regular tree $\Tree_\Delta$ of degree $\Delta$
it is easy to show that
$\lambda_c(\T_\Delta) = (\Delta-1)^{\Delta-1}/(\Delta-2)^\Delta$ \cite{Kelly}.

There are long-standing heuristic results which suggest that
$\lambda_c(\integers^2)\approx 3.796$ \cite{GauntFisher,BET,Racz}.
For the upper bound on the critical activity, a
classical Peierls' type argument implies
$\lambda_c(\integers^2)=O(1)$~\cite{Dobrushin},
and Blanca et al. \cite{BGRT} improved this upper bound to show $\lambda_c(\integers^2)<5.3646$.
Our focus is on the lower bound.

Weitz \cite{Weitz} showed that $\lambda_c(\integers^2) \geq
\lambda_c(\T_4) = 27/16=1.6875$.  His result followed from
the algorithmic result.  For
all graphs with constant maximum degree $\Delta$, $\lambda<\lambda_c(\Tree_\Delta)$,
Weitz~\cite{Weitz} presented an $\fptas$ for approximating the partition function.
A central step in his approach is proving a certain decay of correlations property
known as {\em strong spatial mixing (SSM)} on the graph $G$.
SSM says that for every $v\in V$, every $T\subset V$ and $S\subset T$, and
pair of configurations
$\sigma,\tau$ on $T$ which only differ on $S$ (i.e., $\sigma(T\setminus S)=\tau(T\setminus S))$
then the difference in the influence of $\sigma$ and $\tau$ on
the marginal probability of $v$
decays exponentially in the distance of $v$ from the difference set~$S$
(see Section \ref{sec:prelim} for formal definitions of these concepts).
In contrast, {\em weak spatial mixing (WSM)} only requires that
the influence decays exponentially in the distance to the set~$T$.
For the hard-core model, since fixing a vertex to be unoccupied (or occupied)
can be realized by removing the vertex (or the vertex and its neighbors, respectively),
it then follows that SSM on a graph $G$ is equivalent to WSM for
all vertex induced subgraphs of $G$.

Given $\sigma=(\sigma_v)_{v\in V}$
where $\sigma_v$ is an ordering of the neighbors of $v$,
then Weitz constructs a version of the tree of
self-avoiding walks from $v\in V$ that we will denote as $\Tsaw^\sigma(G,v)$;
see Section \ref{sec:SAW} for its definition.
The key property is that for every $\sigma$,
if for all $v$ SSM holds on $\Tsaw^\sigma(G,v)$ then SSM holds on $G$.
His variant of the self-avoiding walk tree fixes the leaves of the tree (corresponding
to the walk completing a cycle in $G$) to be occupied or unoccupied based
on the ordering $\sigma_v$ for the last vertex $v$ in the corresponding cycle.
He then shows that SSM holds on the complete tree $\Tree_\Delta$,
and hence SSM holds on all trees of maximum degree $\Delta$ when
 $\lambda<\lambda_c(\Tree_\Delta)$.

Restrepo et al.~\cite{RSTVY} improve upon Weitz's approach for $\integers^2$
by utilizing its structure to build a better ``bounding tree''
than $\Tree_\Delta$.  They define a set of branching matrices $\bM_\ell$ for $\ell\geq 4$
corresponding to walks in $\integers^2$ containing no cycles of length $\leq\ell$
(see Section \ref{sec:bm} for a more formal introduction to these notions).
The key point is that, for certain orderings $\sigma$, the tree $\Tsaw^\sigma(\integers^2)$
is a subtree of the tree $T_{\bM_\ell}$
defined by $\bM_\ell$.  They then present a decay of correlation proof by using a suitable message passing approach
for proving SSM for $T_{\bM_\ell}$, and hence for $\Tsaw^\sigma(\integers^2)$ as well.
They show that SSM holds on $T_{\bM_6}$ for $\lambda<2.33$,
and SSM holds on $T_{\bM_8}$ for $\lambda<2.388$.  Consequently, they
establish that $\lambda_c(\integers^2)>2.388$.
In addition, a recent paper \cite{SSY} presents a simpler condition for establishing SSM
based on the connective constant of $\integers^2$, but the bounds obtained by that approach
are currently weaker than~\cite{RSTVY}.

Our first result establishes a limit to these approaches by showing
that SSM does not hold on $\Tsaw(\integers^2)$.  As mentioned earlier, in the
construction of $\Tsaw(\integers^2)$, the assignment for leaves depends on the ordering $\sigma$ which,
for every vertex $v$,
specifies an ordering of the neighbors of $v$.
Since $\integers^2$ is vertex-transitive, it is natural to
define an ordering that is identical for every vertex (e.g., based on an ordering
of the directions $N,S,E$, and $W$), which we refer to as a {\em homogenous ordering}.
Thus, a homogenous ordering $\sigma$ is one where $\sigma_v=\sigma_w$
for all $v,w\in V$.
We prove the following result.
\begin{theorem}
\label{thm:counter}
For all $\sigma$, all $\lambda>3.4$,
SSM does not hold on $\Tsaw^\sigma(\integers^2)$.
Moreover, for all homogenous $\sigma$, all $\lambda>3$,
SSM does not hold
on $\Tsaw^\sigma(\integers^2)$.
 \end{theorem}
The theorem follows
from considering a tree $T$ that is a subtree of $\Tsaw^\sigma(\integers^2)$ and
establishing the threshold for WSM on $T$.
 The tree $T$ that we consider in the homogenous ordering case
is quite simple.  When $N$ is first in the ordering,
the tree is simply the never-go-south tree (see Section~\ref{sec:counter}).
For any $\Tsaw^\sigma(\integers^2)$ that is based on an inhomogeneous ordering $\sigma$,
we are able to find another general subtree for which
the WSM does not hold when $\lambda = 3.4$. Such an example gives
a strong evidence that in order to prove the SSM for $\integers^2$ when $\lambda$ is close to the conjectured threshold,
the self-avoiding walk tree approach might not be appropriate. There are subtrees of the SAW tree of $\integers^2$ that
have lower WSM threshold and hence one has to figure out an approach to exclude such trees.

We then present an improvement of the approach of
Restrepo et al. \cite{RSTVY} for proving SSM for the trees  $T_{\bM_\ell}$.
 They consider a particular statistic of the marginal distributions of the vertices,
  and prove the correlation decay property inductively on the height.
 The statistics can be viewed as a message passing algorithm, a variant of belief propagation.
The messages they consider are a natural
 generalization of the message which is used to analyze the complete tree up to the
 tree threshold $\lambda_c(\Tree_\Delta)$ (which thereby reproves Weitz's result \cite{Weitz}).
 They establish a so-called DMS condition as a sufficient condition for these messages to imply
 SSM holds on the tree under consideration.  Some of the limitations of their approach
 are that to find the settings for
 the parameters in their messages and the DMS condition, they use a heuristic hill-climbing
 algorithm which might become trapped in  local optima.  In addition, verifying their DMS
 condition is non-trivial.

In this paper, we consider piecewise linear functions for the messages.
As a consequence, we can find these functions by solving a linear program.
This yields improved results and simpler proofs of the desired contraction property.
Consequently, we prove SSM holds for $T_{\bM_6}$ when $\lambda\leq 2.45$
(previously, $2.33$ by the DMS condition) and SSM holds for $T_{\bM_8}$ when
$\lambda \leq \newlbd$ (previously, $2.388$).
This establishes the following theorem.
\begin{theorem}
\label{thm:critical}
$\lambda_c(\integers^2) > \newlbd$.
\end{theorem}
The rest of the paper is organized in the following way.
We formally define WSM and SSM in Section \ref{sec:prelim} and also present there the self-avoiding walk tree
construction used by Weitz \cite{Weitz}. In Section \ref{sec:bm}, we will introduce branching matrices and present the framework of Restrepo et al. \cite{RSTVY} in a manner tailored to our work.
 In Section \ref{sec:counter} we will discuss limitations of Weitz's approach by showing several counter-examples. Finally, in Section \ref{sec:lp} we discuss
our linear programming approach for proving SSM,
which yields an improvement on the lower bound for the uniqueness threshold of the hard-core model on $\integers^2$.

\section{Preliminaries}\label{sec:prelim}

\subsection{Definitions of WSM and SSM}
For a graph $G=(V,E)$ and $S\subset V$, we define the boundary condition $\sigma$ on $S$ to be a
fixed configuration on $S$.
For a boundary condition $\sigma$, let $\rp_v(\sigma)$ be the unoccupied probability of vertex~$v$
in the Gibbs distribution $\mu$ on $G$ conditional on $\sigma$.
We now formally define WSM and SSM.

\begin{definition}[Weak Spatial Mixing]\label{def:wsm}
For the hard-core model at activity $\lambda$, for finite graph
$G=(V,E)$, {\em WSM} holds  if there exists $0<\gamma < 1$ such that for every
$v\in V$, every $S\subset V$, and every two configurations
$\sigma_1,\sigma_2$ on $S$,
$$\left|\rp_v(\sigma_1)-\rp_v(\sigma_2)\right|
~\leq~\gamma^{\dist(v,S)}$$ where $\dist(v,S)$ is the graph distance
(i.e., length of the shortest path) between $v$ and (the nearest
point in) the subset $S$.
\end{definition}
For an infinite graph $G$, we define the WSM threshold for $G$ as
\[ WSM(G) = \inf\{\lambda: \mbox{WSM does not hold on $G$ at activity $\lambda$}\}.
\]

\begin{definition}[Strong Spatial Mixing]
For the hard-core model at activity $\lambda$, for finite graph
$G=(V,E)$, {\em SSM} holds if there exists a $0<\gamma < 1$ such that
 for every $v\in V$, every $S\subset V$, every $S'\subset S$,
and every two configurations $\sigma_1,\sigma_2$ on $S$ where
$\sigma_1(S\setminus S')=\sigma_2(S\setminus S')$,
$$\left|\rp_v(\sigma_1)-\rp_v(\sigma_2)\right|
~\leq~\gamma^{\dist(v,S')}.$$
\end{definition}
Finally, let $SSM(G)$ denote the SSM threshold for $G$, defined analogously to $WSM(G)$ but with respect to SSM.

To contrast the definitions of WSM and SSM, note that in WSM the influence decays exponentially in
the distance to the boundary set $S$, whereas in SSM it is exponentially in the distance to the subset $S'$ of the
boundary that they differ on.
An important observation that we repeat from the Introduction to emphasize it, is that
for the hard-core model, for a graph $G$, SSM holds if and only if for all (induced vertex) subgraphs of $G$
WSM holds.

\subsection{Weitz's SAW Tree}\label{sec:SAW}
We now detail Weitz's self-avoiding walk tree construction \cite{Weitz}.
Given $G=(V,E)$, we fix an arbitrary ordering
$\sigma_w$ on the neighbors of each vertex $w$ in $G$.
Let $\sigma=(\sigma_w)_{w\in V}$ be the collection of these orderings.
For each  $v\in V$, the tree $\Tsaw^\sigma(G,v)$ rooted at $v$ is constructed as follows.

Consider the
tree  of self-avoiding walks originating from $v$,
including the vertices closing a cycle in the walks as leaves.
Denote this tree by $\Tsaw(G,v)$.
 We assign a boundary condition to the leaves by the following rule.
 Each leaf closes a cycle in $G$, so say the leaf corresponds to vertex $w$ in $G$
 and the path leading to the leaf corresponds to the path $v \to z_1 \to \dots \to z_j \to w\to v_1 \to \dots \to v_\ell\to w$ in $G$.
 Then if $v_1> v_\ell$ in the ordering $\sigma_w$ we fix this leaf to be unoccupied,
 and if $v_1< v_\ell$ in the ordering $\sigma_w$ we fix this leaf to be occupied.
Since we are in the hard core model, if the leaf is fixed to be unoccupied we simply
 remove that vertex from the tree.  And if the leaf is fixed to be
 occupied, we remove that leaf and all of its neighbors from the tree, i.e.,
 we remove completely the subtree rooted at the parent of that leaf.

If a boundary condition $\Gamma$ is assigned to a subset $S$ of $G$, then the
self-avoiding walk tree can also be constructed consistently to the boundary condition, i.e.,
for a vertex $w\in S$ of $G$, we assign  $\Gamma(w)$
to every occurrence of $w$ in $\Tsaw^\sigma(G,v)$.
  Weitz proves that, for any boundary condition on $G$ and any vertex $v$, the marginal distribution of $v$ on $G$ is the same as the  marginal distribution of the root of $\Tsaw^\sigma(G,v)$ with the corresponding boundary condition. This further implies the following.
\begin{lemma}[Weitz \cite{Weitz}]
\label{thm:weitz}
For a specific $\lambda$, for any $\sigma$,
if for all $v$ SSM holds for $\Tsaw^\sigma(G,v)$,  then SSM holds for $G$.
\end{lemma}

\section{Message Passing Approach for Proving SSM}
\label{sec:bm}
Let us first recall the recurrence of the marginal distributions on trees for the hard-core model.
For now, we fix our infinite tree to be $T$. Let $v$ be a vertex of $T$,
and let $T_v$ denote the subtree of $T$ rooted at $v$. Let $N^-(v)$ denote the children of $v$ in $T_v$.
Let $\alpha_v(\Gamma)$ be the unoccupied probability of vertex $v$ in the subtree $T_v$ rooted at $v$ with boundary condition $\Gamma$.
It is straightforward to establish that $\alpha_v(\Gamma)$ satisfy
the following recurrence:
\begin{equation}
\label{eq:unocc}
\alpha_v{(\Gamma)} = \frac{1}{1 + \lambda \prod_{w\in N^-(v)} \alpha_{w}(\Gamma)}.
\end{equation}
There are two special boundary conditions: one is called the odd boundary condition (denoted as $\Gamma_{o,L}$) which occupies all the vertices at level $L$ when $L$ is odd (and unoccupies when $L$ is even);
 the other is called the even boundary condition  (denoted as $\Gamma_{e,L}$) which occupies all the vertices at level $L$ when $L$ is even (and unoccupies when $L$ is odd).
These two boundary conditions are the extremal ones, meaning that for any other boundary condition $\Gamma$
for the vertices at distance $L$ from the root $r$ of $T$,
$\alpha_r(\Gamma_{e,L}) \le \alpha_r(\Gamma) \le \alpha_r(\Gamma_{o,L})$ when $L$ is even (and with the
inequalities reversed when $L$ is odd).

To see that WSM holds for the tree $T$, it is enough to show that
for the odd and even boundary conditions $\{\Gamma_{o,L}\}_{L\in \mathbb{N}}$ and $\{\Gamma_{e,L}\}_{L\in\mathbb{N}}$, the difference of the marginal probabilities at the root $|\alpha_r(\Gamma_{o,L}) - \alpha_r(\Gamma_{e,L})|$ decay exponentially in $L$.

\subsection{Branching matrices}
Recall that in order to show that uniqueness holds for $\Z^2$ for a certain $\lambda$, it is enough to show that for the same $\lambda$,
SSM holds on a certain tree which is a super-tree of $\Tsaw^\sigma(\Z^2)$.
Due to the regularity of $\Z^2$, in \cite{RSTVY}, deterministic multi-type Galton-Watson trees are proposed to characterize the candidate super-trees.
The trees can be defined by branching matrices in the following way.
A branching matrix $\bM$ is simply a square matrix composed of non-negative integer entries.
\begin{definition}
\label{def:branching-matrix} Given a $t\times t$ branching matrix
$\bM$, $\mathcal F_{\leq\bM}$
is the family of trees which can be generated
under the following restrictions:
\begin{itemize}
\item[$\circ$] Each vertex in tree $T\in \mathcal F_{\leq \bM}$ has its type $i\in \{1,\dots,
t\}$.
\item[$\circ$] Each vertex of type $i$ has at most $M_{ij}$ children of type $j$.
\end{itemize}
\end{definition}
We use $T_{\bM}$ to refer to the tree that is generated by the matrix $\bM$, specifically, we mean the largest tree
in the family $\mathcal F_{\leq \bM}$. The simplest $\bM$ such that $\Tsaw^\sigma(\Z^2)$ is in the family $\cF_{\bM}$
is $\bM = \begin{pmatrix}0&4\\
0&3
\end{pmatrix}$.
In this case, $\T_\bM$ is the complete, regular tree of degree $4$.
Because of the regularity of $\Z^2$, it is clear that $\Tsaw^\sigma(\Z^2)$ is a subtree of
$\T_\bM$.
As shown in \cite{RSTVY},
a more sophisticated set of branching matrices $\bM'$ which contain $\Tsaw(\Z^2)$ in their family
are trees $T_{\bM'}$ corresponding to all walks
of $\Z^2$ truncated when closing a cycle of length less than or equal to a certain constant $\ell$.
Clearly, $T_{\bM'}$ is a super-tree of $\Tsaw{(\Z^2)}$, because any
path in $T_{\bM'}$ will only avoid cycles of length $\leq \ell$ whereas
paths in $\Tsaw^\sigma(\Z^2)$ are avoiding all cycles.

When one tries to avoid a cycle of length $4$, the matrix becomes
$$
\bM'_4=\begin{pmatrix}
0&4&0&0\\
0&1&2&0\\
0&1&1&1\\
0&1&1&0
\end{pmatrix}, $$
where each type is simply representing the various stages of completing a cycle of length $4$ in a walk.
It is easy to verify that $\Tsaw^\sigma(\Z^2)$, for any $\sigma$, is in the family $\cF_{\bM'_4}$.

In $\bM'_4$, we have not yet taken into consideration the effect of the assignments to leaves as
detailed in the construction of $\Tsaw^\sigma$ in Section \ref{sec:SAW}.
 When we do that, we are able to construct much more sophisticated branching matrices
 which yield better bounds.
Therefore, for $\ell\geq 4$, let $\bM_\ell$
denote the branching matrix generating the tree containing all walks in $\integers^2$
truncated when completing a cycle of length $\leq \ell$, where these leaf vertices are occupied or unoccupied
according to the definition in Section \ref{sec:SAW} based on some fixed homogeneous ordering $\sigma$
of the neighbors for every vertex.
By taking into account the boundary condition we obtain a smaller tree
since when a  walk closes a cycle with an occupied assignment to a vertex $u$,
this forces the parent of $u$ to be unoccupied, which further trims down the size of the tree.
These more sophisticated matrices yield a ``tighter'' bound on $\Tsaw^\sigma(\Z^2)$, however
the number of types increase.
For example, for $\ell=4$, whereas $\bM'_4$ has 4 types, $\bM_4$ has 17 types (after
some simplifications), see
\cite{RSTVY} for details of $\bM_4$.  For $\bM_6$ there are 132 types, and for $\bM_8$ there are $922$ types.

\subsection{Contraction Principle}
For each $t$ by $t$ branching matrix $\bM$, we would like to derive a condition such that SSM holds for the tree $T_{\bM}$.
Throughout this paper, for each type $i$, we treat the row $\bM_i$ of $\bM$ as a multi-set and each entry $\bM_i(j)$ of the row denotes the number of elements the set $\bM_i$ has of type $j$.
We use $t(w)$ to denote the type of vertex $w \in \bM_i$.
The following lemma, which is re-stating Lemma 1 from \cite{RSTVY} in
a slightly simpler form that is more convenient for our work,
provides a sufficient condition for SSM to hold for the tree $T_{\bM}$. The proof is in Section \ref{sec:A}.
\begin{lemma}
\label{lem:con}
Let a branching matrix $\bM$ be given. Assume there is $0<\gamma<1$ such that
for each type $i$, there is a positive integrable function $\Psi_i$ where
 \begin{equation}
\label{eq:func}
\frac{1-\alpha_i}{ \Psi_i(\alpha_i)} \sum_{w \in \bM_i} \Psi_{t(w)}(\alpha_w) < \gamma,
\end{equation}
for $\alpha_w$ in the range $[1/(1+\lambda),1]$ for each child $w$ and $\alpha_i = \left(1+\lambda \prod_{w\in \bM_i}\alpha_w\right)^{-1}$ defined in \eqref{eq:unocc} as a function of $\alpha_w$'s.
Then SSM holds for $T_\bM$, i.e., WSM holds for all trees $T$ in the family $\cF_{\le \bM}$ with a fixed rate $\gamma < 1$.
\end{lemma}

\subsection{Reduction of the branching matrices $\bM_\ell$}
\label{sec:red}
Usually, when one applies various methods trying to solve the functional inequality \eqref{eq:func}, one has to face the fact that the dimension of the matrix $\bM$ is huge, e.g., $t = 922$ for $\ell=8$ in \cite{RSTVY}.
A natural way to generate
$\bM$ is using a DFS program that enumerates all of the types by remembering the history of the self-avoiding walk. However, there are many types in such a matrix that are essentially the ``same''. Here we provide a rigorous definition of what types are the same and can be reduced, and a heuristic approach for efficiently
finding those types that are the same.

Let $\C$ be a partition of the types in $\bM$, i.e., $\C = \{C_1,C_2,...,C_k\}$ such that $\biguplus_{i=1}^{k} C_i = [t]$.
%The partition defines naturally a mapping $f_\C :[t] \rightarrow [k]$ taking each type $i \in [t]$ to the class $f_\C(i)$ in $[k]$.
We define the partition to be consistent with $\bM$, if for every $i \in [k]$, each pair of types $s,s' \in C_i$, the rows $\bM_s$ and $\bM_{s'}$ are the same with respect to $\C$, that is
\begin{equation}
\label{eq:consistent}
\sum_{j\in C_{i'}} \bM_{sj} =  \sum_{j\in C_{i'}} \bM_{s'j}, \text{ for all } i'\in [k].
\end{equation}
\begin{definition}\label{def:reducible}
Given $\bM$ and a partition $\C$ of size $k$ which is consistent, we define the $k$-by-$k$ matrix $\bM^{\C}$ by,
\begin{equation}
\label{eq:reduce}
\bM^{\C}_{ii'}= \sum_{j\in C_{i'}} \bM_{sj}\text{ where  }s \in C_i,
\end{equation}
by \eqref{eq:consistent} the choice of $s\in C_i$ does not matter.

We say $\bM$ is {\em reducible} to a $k$-by-$k$ matrix $\bB$ if there is a consistent partition $\C$ such that $\bB = \bM^{ \C}$.
\end{definition}
%The definition of $\bM^{\vdash \C}$ is well-defined by the mapping $f_\C$.
\begin{lemma}
\label{lem:red}
For a branching matrix $\bM$ which is reducible to a matrix $\bB$,
 $$\cF_{\le \bB} = \cF_{\le \bM} ~~\textrm{ and }~~ T_{\bB} = T_{\bM}.$$
\end{lemma}

Note that since the trees $T_{\bB}$ and $T_{\bM}$ are the same (when one ignores the types),
WSM holds on one iff it holds on the other, and similarly for SSM.  Hence, if we can find a small matrix $\bB$ that $\bM$
is reducible to then we can use $\bB$ to simplify proofs of associated spatial mixing properties.

\begin{proof}
Consider a tree $T\in\cF_{\le \bM}$.  For each vertex in $T$ relabel it by its corresponding type in $\C$.
In other words, if vertex $v\in T$ has type $s\in [t]$ then relabel it to type $i\in [k]$ where $s\in C_i$.
By \eqref{eq:consistent} and \eqref{eq:reduce}, this tree $T$ can be generated by $\bB$ and hence $T\in\cF_{\le \bB}$.

For the reverse mapping, consider a tree $T\in \cF_{\le\bB}$.
First, for the root $v$ of type $i\in [k]$, reassign it an arbitrary type $s\in C_i$.
For a vertex $w$, given its new label $s\in [t]$, by \eqref{eq:consistent} and \eqref{eq:reduce}, the
number of children of $w$ that are of type $i'\in [k]$ is $\leq  \sum_{j\in C_{i'}} \bM_{sj}$.
Hence, we can relabel these children with types in $C_{i'}$ so that they are consistent with row $s$ of $\bM$.
After fixing such a relabeling of the children of $w$, then we continue to the children of $w$.  This method relabels the vertices of $T$ to
types in $[t]$ so that the new labeling can be generated by $\bM$ and hence $T\in\cF_{\le \bM}$.
This proves $\cF_{\le \bM} = \cF_{\le \bB}$ and an identical approach shows that $T_{\bB} = T_{\bM}$.
\end{proof}

Now the question is how to find a small consistent partition $\C$. For a specific value $\lambda < \textrm{WSM}(T_{\bM})$, let $V_\lambda$ be the fixed points of the recurrences of the marginal distributions defined by $\bM$. Our conjecture is the following.
\begin{conjecture}\label{conj:red}
Let the partition $\C(\lambda)$ be the sets of types that have the same value of the fixed points in $V_\lambda$, i.e., for each $C_i \in \C(\lambda)$, for all $c \in C_i$, $V_\lambda(c)$ are the same. If for all $\lambda$, the partitions $\C(\lambda)$ are identical, then $\C$ is a partition that is consistent of $\bM$.
\end{conjecture}
Using the intuition from Conjecture \ref{conj:red} we are able to find good partitions in practice. We simply run a dynamic programming algorithm on the tree $T_{\bM}$ to calculate an approximation of the fixed points in $V_\lambda$. Once the approximation is good enough, we simply make the partition according to this approximation. We then check the consistency of the partition with $\bM$, and therefore, we know whether the resulting matrix generates the same tree as the original one or not by Lemma~\ref{lem:red}.
Applying this reduction to $\bM_6$ the number of types
goes down from $132$ to $34$, and for $\bM_8$ the number of types goes down from $922$ to $162$.
This significant reduction in the size of the matrices greatly reduces the number of constraints and variables in our
linear programming formulation.  We will use this technique to simplify the branching matrix $\bD_G$
considered in Section \ref{sec:inhomog} for proving Theorem \ref{thm:counter},
and reduce the matrix from 7 types to 3 types.

\section{Upper Bound on the SSM Threshold}
\label{sec:counter}

As described in the introduction, previous approaches for lower bounding $\lambda_c(\integers^2)$ are based on
proving SSM for $\Tsaw^\sigma(\integers^2)$ for some $\sigma$.
To provide a bound on the strength of these approaches we upper bound the SSM threshold for
$\Tsaw^\sigma(\integers^2)$.
We will show that for $\lambda \geq 3.4$,
for all $\sigma$, SSM does not hold for $\Tsaw^\sigma(\integers^2)$.
We also show that for all $\lambda\geq 3$, for any homogeneous ordering $\sigma$,
SSM does not hold for $\Tsaw^\sigma(\integers^2)$.
Note that these results do not imply anything about WSM/SSM on $\integers^2$,
they simply show a limitation on the power of the current proof approaches.

To prove that SSM does not hold on $\Tsaw^\sigma(\integers^2)$ we define a tree $T$ that is
a subtree of $\Tsaw^\sigma(\integers^2)$ and prove that WSM does not hold on $T$ for sufficiently large $\lambda$.

\subsection{Upper Bound for Homogenous Ordering}
\label{sec:homogenous}

We define a branching matrix $\bD_H$ such that $T_{\bD_H}$ corresponds to the never-go-South tree, and prove
that WSM does not hold on this tree when $\lambda>3$.

Since we are assuming a homogeneous ordering $\sigma$, without loss
of generality assume that North is smallest in the ordering.  We construct $\bD_H$ by considering those walks on
$\integers^2$ that only go North, East, and West.
The branching rules can be written in the following finite state machine way:
\[
0. \ O \rightarrow N ~|~ E ~|~ W, \ \ \ 1. \ N\rightarrow N ~|~ E ~|~ W, \ \ \
2. \ E\rightarrow N ~|~ E, \ \ \
3. \ W\rightarrow N ~|~ W,
\]
where $O$ corresponds to the origin and is a transient state so can be ignored when analyzing
the recurrence.
The branching matrix corresponding to the above rule is
\begin{equation}
\label{eq:m1}
D_H =\begin{pmatrix}
1&1&1\\
1&1&0\\
1&0&1
\end{pmatrix},
\end{equation}
where rows/columns 1, 2, and 3 correspond to North, East, and West respectively.

\begin{lemma}\label{prop:Never-South} Let $\sigma$ be homogenous
ordering where $North$ is the smallest in the order.
The tree $T_{\bD_H}$ generated by the branching matrix $\bD_H$ is a subtree of $\Tsaw^\sigma(\integers^2)$.
\end{lemma}

\begin{proof}
In Weitz's construction (as we presented in Section \ref{sec:SAW}), recall that
$\Tsaw(\integers^2)$ denotes the tree of self-avoiding walks of $\integers^2$ originating from the origin,
including the vertices closing a cycle in the walks as leaves (i.e., in $\Tsaw(\integers^2)$ we have not yet
fixed the leaves to be occupied or unoccupied based on the ordering $\sigma$).
The tree $T_{\bD_H}$ consists of all those self-avoiding walks that never go South, and thus,
it is a subtree of $\Tsaw(\integers^2)$.

Now, in the second part of Weitz's construction, some vertices are deleted from
$\Tsaw(\integers^2)$ to obtain $\Tsaw^\sigma(\integers^2)$.
We need to show that no vertex from $T_{\bD_H}$ is deleted.
A vertex is deleted from $\Tsaw(\integers^2)$ because (i) it is an occupied leaf, or (ii) it is the parent of an occupied leaf.
For a leaf $\eta$ in $\Tsaw(\integers^2)$,
the path to $\eta$ in $\Tsaw(\integers^2)$ corresponds to a walk in $\integers^2$ finishing in a cycle.  Since
there are no cycles in the walks corresponding to $T_{\bD_H}$,
we know that $\eta$ does not appear in $T_{\bD_H}$.  This handles case (i). To handle case (ii),
consider a vertex $\tau$ in $T_{\bD_H}$ which is the parent of a leaf $\eta$ in $\Tsaw(\integers^2)$.
As $\eta$ is a leaf, it corresponds to a path finishing a cycle in $\integers^2$, say
the path is $\P_{\eta}= z_1\to \dots \to z_j \to w\to v_1 \to \dots \to v_\ell\to w$. Then, since $\tau$ is the parent of $\eta$
it corresponds to the path $\P_{\tau} = z_1\to \dots \to z_j \to w\to v_1 \to \dots\to v_\ell$. In $\tau$ there are no South moves,
thus, to close a cycle, the edge $v_\ell\to w$ must be a South move.
Therefore in the boundary condition, $\eta$ is fixed to be unoccupied,
as $v_1> v_\ell$ in the ordering $\sigma_w$, because $v_{\ell}$ is at the North of $w$, and  North is smallest in the ordering, by assumption.
\end{proof}

For the tree $T_{\bD_H}$ we can establish its WSM threshold
as stated in the following result.
\begin{lemma}
\label{thm:three}
$$WSM(T_{\bD_H}) = 3.$$
\end{lemma}

The second part of Theorem \ref{thm:counter} concerning homogenous orderings
follows as an immediate corollary of Lemmas \ref{prop:Never-South} and \ref{thm:three}.

\begin{proof}[Proof of Lemma \ref{thm:three}]
Using the partition $\{N\},\{E,W\}$ of $\{N,E,W\}$ the matrix $\bD_H$ is reduced,
as defined in Definition \ref{def:reducible},
to the following $2\times 2$ matrix:
\begin{equation}
\label{eq:m2}
\bB =\begin{pmatrix}
1&2\\
1&1
\end{pmatrix}.
\end{equation}
From Lemma \ref{lem:red} the matrices $\bB$ and $\bD_H$ generate the same family of trees.
Now the recurrences for the marginal distributions of both types derived from \eqref{eq:unocc} are
\begin{equation}
\label{eq:xy}
F(x,y) = \left( \frac{1}{1 + \lambda x y^2},  \frac{1}{1 + \lambda xy}\right).
\end{equation}
Using some algebra, we are able to determine the fixed points of $F(x,y)$ for $\lambda > 1$
\begin{equation*}
(x_0,y_0) = \left(x_0(\lambda), y_0(\lambda)\right) = \left(\frac{4\lambda + \sqrt{8\lambda + 1} - 1}{8\lambda}, \frac{\sqrt{8\lambda +1}-3}{2(\lambda -1)}\right).
\end{equation*}
We just need to check the eigenvalues of the Jacobian of the recurrences at the fixed point, see e.g., \cite{Robinson}: If the largest eigenvalue is greater than~$1$, then the function around the fixed point is repelling and hence it is impossible for the boundary conditions to converge to this unique fixed point.
If the largest eigenvalue is strictly less than $1$, the function is contracting to the fixed point in its neighborhood.
The Jacobian at the fixed point $(x_0,y_0)$ is the following:
\begin{equation}
\label{eq:jac}
J(\lambda)
=
\begin{pmatrix}
\lambda x_0^2 y_0^2 &2\lambda y_0 x_0^3\\
\lambda y_0^3 &\lambda x_0y_0^2
\end{pmatrix}
.
\end{equation}
Denote the trace of $J(\lambda)$ as $\trz(J(\lambda))= \lambda x_0 y_0^2(x_0+1)$  and
its determinant as $\det(J(\lambda)) = -\lambda^2 x_0^3 y_0^4$.
 The largest eigenvalue of $J(\lambda)$ is then
\[
\rho(\lambda) = \frac{\trz(J(\lambda))}{2} + \left(\frac{\trz(J(\lambda))^2}{4} - \det(J(\lambda))\right)^{1/2},
\]
It is easy to check that $\rho(\lambda)$ is increasing and that for $\lambda =3$, $x_0(3) = 2/3$, $y_0(3) = 1/2$ and $\rho(3) = 1$.
\end{proof}

\subsection{Ordering-Independent Subtree for $\Tsaw$}
\label{sec:inhomog}
In this section, we will define a branching matrix $\bD_G$ such that the generated tree $T_{\bD_G}$ is a
subtree of $\Tsaw^\sigma$ independently on the ordering $\sigma_w$ of edges for each vertex $w$.
This new tree $T_{\bD_G}$ never goes South, and in particular it is a subtree of the tree $T_{\bD_H}$ defined in the previous section. This new tree has further structure to ensure that its leaves are at least distance two from the leaves of the self-avoiding walk tree $\Tsaw(\integers^2)$,
which implies that $T_{\bD_G}$ is a subtree of $\Tsaw^\sigma(\integers^2)$ for every boundary condition $\sigma$.
To achieve this property, we add to the never-go-South tree the rule that if
the walk goes East there must be at least two North steps before it goes West (and similarly, for West to East).
To achieve this we need to remember the last two steps.

The tree is constructed by the following rules:
\[
\begin{array}{lll}
&0.\  O \rightarrow N~|~E~|~W,
\\
1.\  N \rightarrow NN~|~NE~|~NW,
&2.\   W \rightarrow WN ~|~ WW,
&3.\   E \rightarrow EN ~|~ EE,
\\
4.\ NN \rightarrow NN ~|~ NE ~|~ NW,
&5.\ NW \rightarrow WN ~|~ WW,
&6.\ NE \rightarrow EN ~|~ EE,
\\
7.\ WW \rightarrow WN ~|~ WW,
&8.\ EE \rightarrow EN ~|~ EE,
&9.\ WN \rightarrow NW ~|~ NN,
\\
10.\ EN \rightarrow NE ~|~ NN.
\end{array}
\]
Here the state $O$ corresponds to the origin, while $E$, $W$ and $N$ correspond to the first edges in the path. Then each of the states corresponds to the last two visited edges. Notice also that states $O$, $E$, $W$ and $N$ are transient states.
We denote the branching matrix for this tree as $\bD_G$.

\begin{theorem}\label{prop:general}
Let $T_{\bD_G}$ be the tree generated by the branching matrix $\bD_G$, and let $\Tsaw(\integers^2)$ be the
self-avoiding walk tree of $\integers^2$ as defined in Section \ref{sec:SAW}.
Let $\sigma$ be an arbitrarily chosen boundary condition for the leaves of $\Tsaw(\integers^2)$ and let
$\Tsaw^\sigma(\integers^2)$ be the reduced tree (removing leaves unoccupied in $\sigma$ and
for occupied leaves, removing the leaves and their parents).  Then,
$\T_{\bD_G}$ is a subtree of $\Tsaw^\sigma(\integers^2)$.
\end{theorem}

\begin{proof}
A leaf vertex $\eta$ in $\Tsaw^\sigma(\integers^2)$ corresponds to a path $\P_\eta$ which closes a cycle in $\integers^2$.
Thus, $\P_\eta=v_0\to v_1\to \dots\to v_i\to \dots\to v_s\to \dots\to v_{s+t}$ where $v_0=O$ is the origin, all $v_j$ for $j<s+t$ are distinct
and $v_i=v_{s+t}$ for some $i$.
In contrast, we claim that the following property ($\star$) holds for the tree $T_{\bD_G}$: for any vertex $\tau$ in $T_{\bD_G}$,
$\tau$ corresponds to a path $\P_\tau$ which is distance at least two from closing a cycle in
$\integers^2$ (where distance means the minimum number of edges).
In other words, $\P_\tau = v_0\to v_1\to \dots \to v_s$  where $v_0 = O$ is the origin, all $v_j$ for $j\leq s$ are distinct and, for all $i<s-1$,
$\dist_{\integers^2}(v_i,v_s) \geq 2$.

Let us first prove this property ($\star$). Consider a vertex $\tau$ in $T_{\bD_G}$, and let $\P_\tau = v_0\to v_1\to \dots\to v_s$
be the corresponding path in $\integers^2$.
Suppose $v_s$ is distance 1 from some $v_i$ for $i < s-1$, then $\P' = v_i\to  \dots \to  v_s \to v_i$ is a cycle in $\integers^2$.
This cycle must contain at least one South move, but by the construction of $T_{\bD_G}$ there are no South moves in $\P_\tau$.
Thus $v_s\to v_i$ is the unique South move in the cycle $\P'$. In $\P'$ since there is only one South move,
there must be exactly one North move, and the rest of moves should be East and West moves.
There must be at least one East and one West move,
but this contradicts that by the definition of $T_{\bD_G}$ we know that in $\P_\tau$
between any East and West moves there are at least two North moves.
This completes the proof of ($\star$).

Now using ($\star$) we can complete the proof of the lemma. Recall that $\Tsaw(\integers^2)$ is the tree of self-avoiding walks on $\integers^2$
before we assign a boundary condition to its leaves based on $\sigma$.
Note that $T_{\bD_G}$ is a subtree of~$\Tsaw(\integers^2)$.
Consider an arbitrary vertex $\tau$ in $T_{\bD_G}$. If $\tau \notin \Tsaw^\sigma(\integers^2)$, then $\tau$ is either a leaf or the parent of (an occupied) leaf in $\Tsaw(\integers^2)$. To finish the proof it is then enough to show that ($\star$) implies that neither a leaf nor the parent of a leaf in $\Tsaw(\integers^2)$ is in $T_{\bD_G}$. Let $\eta$ be a leaf in $\Tsaw(\integers^2)$. Then  $\eta$ corresponds to a path $\P_\eta=v_0\to v_1\to\dots\to v_s \to w \to v_{s+1}\to \dots\to v_{s+k} \to w$. Then the path $\P_\eta$ is distance 0 from closing a cycle in $\integers^2$, which contradicts ($\star$), and thus $\eta$ is not in $T_{\bD_G}$.
Now, let $\eta'$ be the parent of $\eta$. The corresponding path is $\P_{\eta'} = v_0\to v_1\to\dots\to v_s \to w \to v_{s+1}\to \dots\to v_{s+k}$ and $\dist_{\integers^2}(v_{s+k},w) = 1 < 2$. This again contradicts ($\star$) and thus $\eta'$ is not in $T_{\bD_G}$.
\end{proof}

We establish the following bounds on the WSM threshold for the tree~$T_{\bD_G}$.

\begin{lemma}
\label{lem:D1}
For the tree $T_{\bD_G}$, at $\lambda=3.4$ WSM does not hold.
\end{lemma}

%This lemma is proved in Section \ref{sec:proof-inhomogenous}.
Theorem \ref{prop:general} implies that $T_{\bD_G}$ is a subtree for $\Tsaw^\sigma(\integers^2)$ for
any ordering $\sigma$.  Hence, the first part of Theorem \ref{thm:counter} concerning arbitrary orderings $\sigma$
follows as an immediate corollary of Theorem \ref{prop:general} and Lemma \ref{lem:D1}.

To show that the bound in Lemma \ref{lem:D1} is reasonably close to tight we also
show that at $\lambda=3.3$ WSM holds.

\begin{lemma}\label{Lem:lower33} Let $\lambda = 3.3$. Then
WSM holds for $T_{\bD_G}$.
\end{lemma}

\subsection{Proofs of Lemmas \ref{lem:D1} and \ref{Lem:lower33}}

 Using the method introduced in Section \ref{sec:red}, with the partition $\{NN\}$,$\{NW,NE,WW,EE\}$,$\{WN,EN\}$
 of the set of states, the branching matrix $\bD_G$ is reduced to the three-state matrix, $\begin{bmatrix} 1 & 2 & 0 \\ 0 & 1 & 1\\ 1 & 1 & 0\end{bmatrix}$.
Therefore, the recurrences for the marginal distributions of the 3 types are
\[
x \mapsto
F(x) = \left(
\frac{1}{1+\lambda x_1 x_2^2}, \frac{1}{1+\lambda x_2x_3}, \frac{1}{1+\lambda x_1x_2} \right).
\]
Let $F^0(x) = F(x)$ and $F^{n+1}(x) = F(F^n(x))$. Also, for any $x \le y \in \{0,1\}^3$ let $C(x,y) = \{u \in [0,1]^3: x \le u \le y\}$ be the {\it (rectangular) cuboid  defined by} $x$ and $y$. Let $\bar 0 = (0,0,0)$ and $\bar 1 = (1,1,1)$. Then $C(\bar 0,\bar 1) = [0,1]^3$.

We have that WSM holds for $T_{\bD_G}$ if and only $\cap_{n=0}^\infty F^n([0,1]^3)$ is a singleton $\{x^*\}$, which in particular implies that $x^*$ is the (unique) fix point of $F$ in $[0,1]^3$.

We use numerical computations as part of our proof. In order to do exact computations we use $F^\uparrow$ and $F^\downarrow$ functions approximating $F$ to $7$ decimal digits. To do this we define $S = \{0,10^{-7},2*10^{-7},\dots, 1\}$ and $F^\downarrow, F^\uparrow: S \to S$ by $F^\downarrow(x) = \lfloor{ F(x)*10^{7}\rfloor}*10^{-7}$ and $F^\uparrow(x) = \lceil{ F(x)*10^{7}\rceil}*10^{-7}$. We have then for any $x \in S$,
\begin{equation}\label{eq:boundF}
F^\downarrow(x) \le F(x) \le F^\uparrow(x).
\end{equation}
Using monotonicity, and induction it follows that for any $n$, and $x$
\begin{equation}\label{eq:boundFn}
(F^\downarrow)^{2n}(x) \leq F^{2n}(x) \leq  (F^\uparrow)^{2n}(x).
\end{equation}
Notice also that for any $x \in S$, $F^\downarrow(x), F^\uparrow(x)$ can be computed using exact arithmetic.

\begin{proof}[Proof of Lemma \ref{lem:D1}]
Fix $\lambda = 3.4$. Let
\begin{align*}
 &x_{LL} = [0.5483975, 0.4870566, 0.4178331]\qquad& x_{LR} = [0.5489575, 0.4874566, 0.4182131]\\
 &x_{RL} = [0.6927559, 0.5906225, 0.5236103] & x_{RR} = [0.6933359, 0.5910425, 0.5240703].
\end{align*}
Let $C_L = C(x_{LL},x_{LR})$ and $C_R = C(x_{RL},x_{RR})$.
  We check numerically that
  \[
  x_{LL} < F^\downarrow(x_{RR}) < F^\uparrow(x_{RL}) < x_{LR} \text{ and }  x_{RL} < F^\downarrow(x_{LR}) < F^\uparrow(x_{LL}) < x_{RR}.
  \]
Using monotonicity and \eqref{eq:boundF} we obtain  $F(C_L) = F(C(X_{LL},X_{LR})) \subset C(F(X_{LR}),F(X_{LL})) \subset C(F^\downarrow(X_{LR}),F^\uparrow(X_{LL}))  \subset C(x_{RL},x_{RR}) = C_R$. Similarly $F(C_R) \subset C_L$.

Applying the Brouwer fixed point theorem, $F^2$ has a fixed point $x_L$ in $C_L$ and a fixed point $x_R = F(x_L)$ in $C_R$. As $C_L \cap C_R = \emptyset$, $x_L \neq x_R$. Thus $\cap_{n>0} F^n([0,1]^3) \supset \{x_L,x_R\}$. Therefore WSM does not hold for $\lambda = 3.4$.
\end{proof}

Now we show that WSM holds for $T_{\bD_G}$ when $\lambda = 3.3$ (Lemma \ref{Lem:lower33}). The main idea of the proof is to find $U \subset [0,1]^3$ where $F$ is contracting and such that $F(U) \subset U$. To show that $F$ is contracting in $U$ we will show that $\rho(J_F)$, the spectral radius of the Jacobian of $F$ is strictly smaller than $1$.

\begin{lemma}\label{lem:fixPoint}
Assume $x_L \le F^{2N}(\bar 0)$ and $F^{2N}(\bar 1) \le x_R$, for some $N \ge 0$ and that $\rho(J_F(u)) < 1$ for all $u \in C(x_L,x_R)$. Then $\cap_{n=0}^\infty F^n([0,1]^3)$ is a singleton.
\end{lemma}
\begin{proof} From the Brouwer fixed-point theorem $\cap_{n=0}^\infty F^n([0,1]^3) \neq \emptyset$. Let $C = C(x_L,x_R)$.  From monotonicity, $F^{2N}(x) \in C$ for any $x \in [0,1]^3$. In particular $F^{2N}(C) \subset C$. As $F$ is contracting in $C$, $F^{2N}$ is also contracting on $C$ and from Banach's fixed point theorem, $\cap_{n=0}^\infty F^{2nN}(C)$ is a singleton. Now let $y_1,y_2 \in \cap_{n=0}^\infty F^n([0,1]^3)$. Then $F^{k}(y_1),F^{k}(y_2) \in \cap_{n=0}^\infty F^n([0,1]^3)$ for any $k \ge 0$. By assumption, $F^{2N}(y_1),F^{2N}(y_2) \in C$. Therefore,
$F^{4N}(y_1),F^{4N}(y_2) \in \cap_{n=0}^\infty F^{2nN}(C)$ and thus $F^{4N}(y_1) = F^{4N}(y_2)$. As $F$ is one-to-one then $y_1 = y_2$.
\end{proof}

Notice that the Jacobian of $F$ is
\[
J_F(x) = -\lambda
\begin{bmatrix}
x_2^2F^2_1(x)& 2x_1x_2F^2_1(x)&  0 \\
    0&     x_3F^2_2(x) & x_2F^2_2(x) \\
   x_2F^2_3(x)&     x_1F^2_3(x)&  0 \end{bmatrix}.
\]
In order to bound the spectral radius of $J_F(x)$ we will compute the spectral radius of the  matrix
\[M(x) = \begin{bmatrix}
x^2_1x_2^2& 2x_1^3x_2&  0 \\
    0&     x^2_2x_3 & x^3_2 \\
   x_2x^2_3&     x_1x^2_3&  0
   \end{bmatrix}.
\]
\begin{lemma}\label{lem:boundrho} Let $x,y \in [0,1]^3$ be such that $F(x) \le y$. Then
for any $u \in C(x,y)$ we have $\rho(J_F(u)) \le \lambda \rho(M(y))$.
\end{lemma}
\begin{proof} For any $u \in C(x,y)$ we have $F(u) \le F(x) \le y$. Thus for any $i,j \in \{1,2,3\}$,
$0 \leq -J_F(u)_{ij} \le \lambda M(y)_{ij}$.  Thus (see \cite{spec2}), $\rho(J_F(u)) = \rho(-J_F(u)) \le \rho(\lambda M(y)) = \lambda \rho( M(y))$
\end{proof}

\begin{proof}[Proof of Lemma \ref{Lem:lower33}]
 Let $N = 10^3$. Let  $x_L = (F^\downarrow F^\uparrow)^N(\bar 0) = [0.6234082, 0.5418325, 0.4728517]$ and  $x_R = (F^\uparrow F^\downarrow)^N(\bar 0) = [0.6234525, 0.5418642, 0.4728841]$. We check numerically that $(F^\downarrow)(x_R) = x_L$ and $(F^\uparrow)(x_L) = x_R$. Using \eqref{eq:boundFn} we obtain
\begin{equation}\label{eq:numCheck}
x_L \le F(x_R) \le F^{2N}(\bar 0) \le  F^{2N}(\bar 1)\le F(x_L) \le x_R.
\end{equation}
From Lemmas \ref{lem:fixPoint} and \ref{lem:boundrho} it is enough to show that $\lambda \rho (M(x_R)) < 1$.
Let $v = (0.685, 0.49, 0.5)^T$, then (see \cite{spec1,spec2})
$\lambda \rho(M(x_R)) \leq \lambda \max_{i =1,2,3} \frac {(M(x_R)v)_i}{v_i} < 0.9998$, where the last inequality is checked numerically using exact arithmetic.
\end{proof}

\subsection{Tree with Different Thresholds for SSM and WSM}
\label{sec:WSMneqSSM}

Brightwell et al.~\cite{BW99} give an example of a tree for which WSM holds but SSM does not hold
for the same activity $\lambda$. Here, we present another example which is more closely
related to $\Tsaw^\sigma(\integers^2)$.  We
show a tree $T'$, which is a super-tree of $T_{\bD_H}$ and subtree of $\Tsaw^\sigma(\integers^2)$
for homogenous $\sigma$,
for which WSM holds
for some $\lambda>3$.

To construct the tree $T'$ we allow some South moves in the tree in a certain context.
In particular, we only allow that a South move happens when the path contains the following substring: NNEESEEN, i.e., a South move is allowed if and only if it is after a sequence of NNEE moves and followed by EEN moves. Also, after the substring NNEESEEN we only allow
the next move to be a North or East move.
Symmetrically we allow paths containing the substring NNWWSWWN followed by a North or West move.

Before formally defining the branching matrix it is useful to discuss the important property
of this tree $T_{\bD'}$.  In the never-go-South tree $T_{\bD_H}$ to complete a cycle one needs at least
one additional South move.  The tree $T_{\bD'}$ includes South moves but it is defined in such a way
that to complete a cycle
one either needs an additional South move or at least 2 additional moves (such as WW).

%To see why this property holds, suppose we want to utilize a South move in $T_{\bD'}$ to complete a cycle.
%Suppose this South move is contained in the sequence NNEESEEN.  After this sequence, before going West we first need to make an
%additional North move.  So there are at least 2 North moves before a West move.  Hence to complete a cycle
%using this South move we need at least 1 additional South move.  Similarly, if the South move is contained
%in the sequence NNWWSWWN, there are at least 2 North moves before an East move, and hence again
%an additional South move is needed to complete a cycle.

The tree $T_{\bD'}$ is formally defined by the following branching matrix denoted as~$\bD'$.
The tree family can be formalized in the following finite state machine way:

\begin{tabular}{lll}
0. $O \rightarrow N ~|~ E ~|~ W$,
&5. $NNE\rightarrow N ~|~ NEE$
& 11.  $NNW\rightarrow N ~|~ NWW$
\\
 1.  $N\rightarrow E ~|~ W ~|~NN$
& 6.  $NEE \rightarrow N ~|~ E ~|~ EES$
& 12.  $NWW \rightarrow N ~|~ W ~|~ WWS$
\\
 2.  $E\rightarrow N ~|~ E$
& 7.  $EES \rightarrow ESE$
& 13.  $WWS \rightarrow WSW$
\\
3.  $W\rightarrow N ~|~ W$
& 8.  $ESE \rightarrow SEE$
& 14.  $WSW \rightarrow SWW$
\\
 4.  $NN\rightarrow NN~|~ NNE~|~ NNW$
& 9.  $SEE \rightarrow EEN$
& 15.  $SWW \rightarrow WWN$
\\
& 10.  $EEN \rightarrow N ~|~ E$
& 16.  $WWN \rightarrow N ~|~ W$
\end{tabular}

Let the matrix describing the above rules be denoted as $\bD'$.
%Other possibilities here: Allow NNE ... ESE ... ENE .... ESE .... N ...
\begin{lemma}
Let $\sigma$ be homogenous
ordering where North is the smallest in the order.
The tree $T_{\bD'}$ that is generated by $\bD'$ is a super-tree of $T_{\bD_H}$ and is a subtree of $\Tsaw^\sigma(\Z^2)$.
\end{lemma}
\begin{proof}
Since $\bD_H$ is a subset of $\bD'$ we have that $T_{\bD'}$ is a super-tree of $T_{\bD_H}$.
As in the proofs of Lemma \ref{prop:Never-South} and Theorem \ref{prop:general},
since $T_{\bD'}$ contains a subset of the self-avoiding walks in $\integers^2$, we again have
that $T_{\bD'}$ is a subtree of $\Tsaw(\Z^2)$. It remains to handle the boundary assignment to $\Tsaw(\Z^2)$, and this is done in a manner similar
 to the proof of Lemma \ref{prop:Never-South}.

As the boundary assignment  to $\Tsaw(\Z^2)$ could only remove leaves and parents of leaves from it, any vertex in
$T_{\bD'}$  which is distance two or more from a leaf of $\Tsaw(\Z^2)$ is a vertex in $\Tsaw^\sigma(\Z^2)$. Also no leaf in $\Tsaw(\Z^2)$ is in $T_{\bD'}$.
Thus, if $\tau \in T_{bD'}$ is such that $\tau \notin \Tsaw^\sigma(\Z^2)$, then there is  a leaf $\eta$ of $\Tsaw(\Z^2)$ such that $\tau$ is the parent of $\eta$ (in $\Tsaw(\Z^2)$) and $\eta$ is set to occupied in $\Tsaw^\sigma(\Z^2)$.

Let the path corresponding to $\eta$ be  $\P_\eta=v_0\to v_1\to \dots\to v_i\to \dots\to v_s$ where $v_0=O$ is the origin, all $v_j$ for $j<s$ are distinct and $v_i=v_s$ for some $i < s-3$. The path corresponding to $\tau$ is then $\P_\tau=v_0\to v_1\to \dots\to v_i\to \dots\to v_{s-1}$. The order $\sigma$ is homogeneous with $N$ as the smallest element, and $\eta$ is set to occupied. Thus the move $v_{s-1} \to v_{s}$ could not be a South move.

Now, when we consider the cycle $v_i\to \dots v_{s-1} \to v_{s}$. This cycle contains the same number of South and North moves. And thus the path $\P = v_i\to \dots v_{s-1}$ contains at least as many South moves as North moves. If it does not contain a South move, then it does not contain a North move, and one move is not enough to close a cycle. Thus it must contain a South move. Lets assume this South move occurs at a sequence $NNEESEEN$ as the other case is symmetric.  The length of $\P$ is  at least three, and thus it must contain at least two $E$'s from the sequence $NNEESEEN$. Thus $\P$ must contain at least one $W$. But to switch from $E$ to $W$ the number of $N$'s should be larger than the number of $S$'s in $\P$ which is a contradiction.
\end{proof}

We will prove that the WSM threshold
for $T'= T_{\bD'}$, the tree generated by $\bD'$ is above $\lambda = 3.1$,
and hence, combined with Lemma \ref{thm:three}, we get the following lemma.

\begin{lemma}
\label{lem:example}
For the tree $T_{\bD'}$ at $\lambda=3.1$,
WSM holds but SSM does not hold.
\end{lemma}
%We prove Lemma \ref{lem:example} following the same approach as used for
%the proof of Lemma \ref{Lem:lower33}.
\begin{proof}%[Proof of Lemma \ref{lem:example}]
Since $T_{\bD'}$ is a super tree of $T_{\bD_H}$, and, by Lemma \ref{lem:D1} we know that $WSM(T_{\bD_H})~=~3$, therefore SSM does not hold for $T_{\bD'}$ when $\lambda = 3.1$.
Hence, it remains to show that WSM for $T_{\bD'}$ holds when $\lambda = 3.1$,
and the proof of this fact will follow the same procedure as in the proof of Lemma \ref{Lem:lower33}.

Using the method introduced in Section \ref{sec:red}, with the partition $\{N\}$, $\{E,W,EEN,WWN\}$, $\{NN\}$, $\{NNE,NNW\}$, $\{NEE,NWW\}$, $\{EES,WWS\}$, $\{ESE,WSW\}$, $\{SEE,SWW\}$
 of the set of states, the branching matrix $\bD'$ is reduced to an eight-state matrix.
\[\begin{bmatrix}
0 & 2 & 1 & 0 & 0 & 0 & 0 & 0\\ 1 & 1 & 0 & 0 & 0 & 0 & 0 & 0\\ 0 & 0 & 1 & 2 & 0 & 0 & 0 & 0\\ 1 & 0 & 0 & 0 & 1 & 0 & 0 & 0\\ 1 & 1 & 0 & 0 & 0 & 1 & 0 & 0\\ 0 & 0 & 0 & 0 & 0 & 0 & 1 & 0\\ 0 & 0 & 0 & 0 & 0 & 0 & 0 & 1\\ 0 & 1 & 0 & 0 & 0 & 0 & 0 & 0
\end{bmatrix}.\]
Fix $\lambda = 3.1$,
% From which 
the recurrences for the marginal distributions of the 8 types are
\begin{align*}
x \mapsto
F(x) = \left(
\frac{1}{1+\lambda x_3 x_2^2},\right.
&\frac{1}{1+\lambda x_1 x_2},
\frac{1}{1+\lambda x_3 x_4^2},
\frac{1}{1+\lambda x_1 x_5},\\
&\left.
 \frac{1}{1+\lambda x_1 x_2 x_6},
 \frac{1}{1+\lambda x_7},
\frac{1}{1+\lambda x_8},
\frac{1}{1+\lambda x_2}
\right).
\end{align*}
Again, we use $F^\uparrow$ and $F^\downarrow$ functions approximating $F$ to $7$ decimal digits. That is, we define $S = \{0,10^{-7},2*10^{-7},\dots, 1\}$ and $F^\downarrow, F^\uparrow: S \to S$ by $F^\downarrow(x) = \lfloor{ F(x)*10^{7}\rfloor}*10^{-7}$ and $F^\uparrow(x) = \lceil{ F(x)*10^{7}\rceil}*10^{-7}$.

Let $N = 10^3$. Let  
\[x_L = (F^\downarrow F^\uparrow)^N(\bar 0) = [0.6403710, 0.5012248, 0.7209949, 0.4160656, 0.7069206, 0.4166175, 0.4516958, 0.3915610]\]
 and  
 \[x_R = (F^\uparrow F^\downarrow)^N(\bar 0) = [0.6404050, 0.5012516, 0.7210239, 0.4160871, 0.7069451, 0.4166221, 0.4517041, 0.3915739].\]
  We check numerically that $(F^\downarrow)(x_R) = x_L$ and $(F^\uparrow)(x_L) = x_R$. Using \eqref{eq:boundFn} we obtain
\begin{equation}\label{eq:numCheck}
x_L \le F(x_R) \le F^{2N}(\bar 0) \le  F^{2N}(\bar 1)\le F(x_L) \le x_R.
\end{equation}

In this case the Jacobian of $F$ is
\[
J_F(x) = -\lambda
\begin{bmatrix}
     0& 2 x_2 x_3 F^2_1(x)& x_2^2 F^2_1(x)&       0&  0&     0& 0& 0\\
    x_2 F^2_2(x)&      x_1 F^2_2(x)&    0&       0&  0&     0& 0& 0\\
     0&       0& x_4^2 F^2_3(x)& 2 x_3 x_4 F^2_3(x)&  0&     0& 0& 0\\
    x_5 F^2_4(x)&       0&    0&       0& x_1 F^2_4(x)&     0& 0& 0\\
 x_2 x_6 F^2_5(x)&   x_1 x_6 F^2_5(x)&    0&       0&  0& x_1 x_2 F^2_5(x)& 0& 0\\
     0&       0&    0&       0&  0&     0& F^2_6(x)& 0\\
     0&       0&    0&       0&  0&     0& 0& F^2_7(x)\\
     0&       F^2_8(x)&    0&       0&  0&     0& 0& 0
 \end{bmatrix}.
\]
In order to bound the spectral radius of $J_F(x)$ we will compute the spectral radius of the  matrix
\[M(x) = 
\begin{bmatrix}
     0& 2 x_2 x_3 x^2_1& x_2^2 x^2_1&       0&  0&     0& 0& 0\\
    x_2 x^2_2&      x_1 x^2_2&    0&       0&  0&     0& 0& 0\\
     0&       0& x_4^2 x^2_3& 2 x_3 x_4 x^2_3&  0&     0& 0& 0\\
    x_5 x^2_4&       0&    0&       0& x_1 x^2_4&     0& 0& 0\\
 x_2 x_6 x^2_5&   x_1 x_6 x^2_5&    0&       0&  0& x_1 x_2 x^2_5& 0& 0\\
     0&       0&    0&       0&  0&     0& x^2_6& 0\\
     0&       0&    0&       0&  0&     0& 0& x^2_7\\
     0&       x^2_8&    0&       0&  0&     0& 0& 0
 \end{bmatrix}.
\]

From Lemmas \ref{lem:fixPoint} and \ref{lem:boundrho} it is enough to show that $\lambda \rho (M(x_R)) < 1$.
Let 
\[v = ( .537, .422, .456, .337, .385, .069, .128, .201)^T,\]
 then (see \cite{spec1,spec2})
$\lambda \rho(M(x_R)) \leq \lambda \max_{i =1,\dots,8} \frac {(M(x_R)v)_i}{v_i} < 0.999$, where the last inequality is checked numerically using exact arithmetic.

\end{proof}

\section{Linear Program for Lower Bounding SSM Threshold}
\label{sec:lp}
Here we propose a way to use linear programming to solve the functional inequality \eqref{eq:func}. Notice that if $\Psi_i$ is positive and bounded for all $i$ then inequality \eqref{eq:func} is equivalent to
 \begin{equation}\label{eq:LPgoal}
(1-\alpha_i) \sum_{w\in \bM_i} \Psi_{t(w)}(\alpha_w) < { \Psi_i(\alpha_i)}.
\end{equation}

The idea to solve \eqref{eq:LPgoal} is simple. We will restrict the search for $\Psi_i$ to a family of positive piecewise linear functions with a finite number of discontinuities.

First of all, it is  a simple fact that each $\alpha_i$ is in the interval $I = \left[{1}/{(1+\lambda)}, 1\right]$. We will divide $I$  into a set of $d$ consecutive sub-intervals of the same size. Define
\[X_k = \frac{1}{1+\lambda} + k\frac{\lambda}{d(1+\lambda)}, \text{ for } k = 0,\dots,d-1.
\]
To ease the notation define $Y_k = X_{k+1}$ for $k=0,\dots,d-1$.
Note that the intervals $[X_k,Y_k]$ partition~$I$. Since the only requirements of $\Psi_i(x)$ are positive and integrable, we restrict the search for $\Psi_i(x)$ to functions of linear form $- a_{i,k} x + b_{i,k}$ in each interval $[X_k, Y_k]$ with $a_{i,k}, b_{i,k} > 0$.

Now, for each type $i$, the functional inequality can be decomposed according to different combinations of the intervals of the variables $\alpha_w$ which are type $i$'s children. For each combination, we are able to write down a set of linear inequalities such that it is a sufficient condition for the functional inequality to hold within that region.

To capture for which sub-intervals should \eqref{eq:LPgoal} hold, we say that $k = (k_0,k_1,k_2,\dots,k_{\Delta_i})$, a tuple of indexes, is $i$-acceptable if the interval $[X_{k_0},Y_{k_0}]$ intersects the interval $\left[\frac{1}{1+\lambda\prod_{j=1}^{\Delta_i}Y_{k_j}}, \frac{1}{1+\lambda\prod_{j=1}^{\Delta_i}X_{k_j}}\right]$. We have the following theorem.

\begin{theorem}
\label{thm:LP}
In order for the functional inequality \eqref{eq:func} to hold, it is enough for the following set of linear constraints ($a$'s and $b$'s are the variables) to be feasible:

For each $i \in [t]$ and each $i$-acceptable tuple $k$,
\begin{equation}
\label{eq:LP1}
 \left(1 - X_{k_0}\right) \sum_{j=1}^{\Delta_i}\left( b_{t(j),k_j} - a_{t(j), k_j} X_{k_j} \right) <  \left(  b_{i,k_0} - a_{i, k_0}Y_{k_0} \right),
\end{equation}
where $\{t(j):j=1,\dots,\Delta_i\} = M_i$ (as multisets).

For each $i \in [t]$ and $k = 0,\dots,d-1$,
\begin{equation}
\label{eq:LP2}
 b_{i,k}-a_{i,k}Y_{k}>0,\qquad  0 \leq  a_{i,k} \leq M \qquad 0 \leq b_{i,k} \leq M.
\end{equation}
where $M$ is some (big) constant.
\end{theorem}

\begin{proof}
 Define $\Psi_i(x) =b_{i,k} -a_{i,k}x $ for all $x \in [X_k,Y_k)$. Linear constraints \eqref{eq:LP2} imply
that $\Psi_i$ is non-negative and bounded. Thus it is enough to show \eqref{eq:LPgoal} holds.

Now fix type $i$ we have $k_w$'s such that $\alpha_w \in [X_{k_w},Y_{k_w}]$ for each $w \in \bM_i$. Let $\alpha_i  = 1/(1 + \lambda \prod_{w\in \bM_i} \alpha_w)$,
then
\[
\frac{1}{1+\lambda\prod_{w\in \bM_i}Y_{k_w}} \leq \alpha_i \leq \frac{1}{1+\lambda\prod_{w\in \bM_i}X_{k_w}}.
\]
 Thus if $k_i$ is such that $X_{k_i} \leq \alpha_i \leq Y_{k_i}$ then the tuple $k = (k_i,k_{w_1},\dots,k_{w_{\Delta_i}})$ is $i$-acceptable.

Therefore,
\begin{align*}
(1-\alpha_i) \sum_{w\in \bM_i} \Psi_{t(w)}(\alpha_w)
&=
(1-\alpha_i)\sum_{w \in \bM_i} \left( b_{t(w),k_w}- a_{t(w),k_w}(\alpha_w) \right)\\
&\leq
 \left(1 - X_{k_i}\right) \sum_{j=1}^{\Delta_i}\left( b_{t(w),k_w}-a_{t(w), k_w} X_{k_w} \right)\\
 &<   b_{i,k_i} - a_{i, k_i}Y_{k_i}  &\text{from \eqref{eq:LP1}}\\
 &\le { \Psi_i(\alpha_i)}.
\end{align*}
\end{proof}

Consider the branching matrix $\bM_{\ell}$ generating the family of trees avoiding cycles of length $\le \ell$.  Recall that the tree $T_{\bM_{\ell}}$ which is generated by $\bM_\ell$ is a super-tree of $\Tsaw^\sigma(\Z^2)$, for homogenous $\sigma$.
 We show that the system \eqref{eq:LP1}-\eqref{eq:LP2} corresponding to $\bM_\ell$ is feasible,
 proving SSM for $T_{\bM_\ell}$ and hence for $\Z^2$.

To solve the feasibility problem, we add a new variable $v$ in the right hand side of each linear constraint $ax \le b$, changing this constraint to $ax - b \le v$. We minimize $v$, which is an upper bound for the maximum violation by $x$ among all constraints. The original linear system is feasible if and only the linear program has optimal solution $v < 0$.

The number of constraints and variables in this LP are huge (almost $10$ billion constraints and $1$ million variables) when $d = 200$ for the matrix $\bM_8$. In order to solve the linear program efficiently, one has significantly to reduce its size. In Section \ref{sec:LP-solve}, we will discuss about the methods we use to solve this LP. When running the linear programs built for $\bM_4$ we obtain $\lambda > 2.31$, and for $\bM_6$ we obtain $\lambda > 2.45$, and for $\bM_8$ we obtain $\lambda > \newlbd$.
In this way, we are able to prove that SSM holds for $\integers^2$ for $\lambda \leq \newlbd$.
The data for these LP solutions are available in our online appendix \cite{online}.

What we obtain from our linear program method are closer to the limit of this approach.
Computational experiments suggest the threshold for WSM for $T_{\bM_4}$
is at roughly $\lambda\approx 2.482$, for $T_{\bM_6}$ at
$\lambda \approx 2.653$, for $T_{\bM_8}$ at $\lambda \approx 2.75$ and
finally for $T_{\bM_{10}}$ at $\lambda\approx 2.82$.
These are thesholds for WSM, and the SSM threshold may in fact be even lower, as occurred for our
example in Section~\ref{sec:WSMneqSSM}.

\subsection{Comparison with Previous Approaches}

This method has several advantages compared to the method that is proposed in \cite{RSTVY} in which
a sufficient condition called the DMS condition, is introduced.
DMS is a nonlinear matrix inequality obtained by comparing the geometric mean with the arithmetic mean when one analyzes the functional inequality \eqref{eq:func} for a specific type of $\Psi_i$ functions. These functions are the optimal ones when the tree $T_\bM$ is a complete regular tree.
However, for multi-type branching matrices, they are not necessarily optimal. One has to find the parameters of these
functions $\Psi_i$ in order to satisfy the DMS condition.
The parameters for the
DMS condition are obtained
by a randomized hill-climbing program which may become trapped in a local optima. In contrast, the linear programming method
we present here provides the optimal solution for the class of functions being considered.

For the SSM threshold of $T_{\bM_\ell}$, our method includes the approximation of a more general class of functions and hence we obtain better lower bounds (see Figure \ref{figure-phii}).
Finally, the mathematical correctness of the linear programming method
is very straightforward to check as compared to checking the correctness of the DMS condition.
For $\ell = 4,6,8$, we summarize in the following table,
the experimental lower bound for the WSM threshold of $\bM_\ell$, the size of the matrix $\bM_\ell$, the lower bounds of the SSM threshold for $\bM_\ell$ obtained from DMS condition in \cite{RSTVY}, and the lower bounds of the SSM
threshold for $\bM_\ell$ obtained from our linear program approach.
\begin{center}
\begin{tabular}{|c|c|c|c|c|}
\hline
$\ell$&WSM threshold&Number of Types &$\lambda$ from DMS in \cite{RSTVY} & $\lambda$ from LP\\
\hline\hline
4&2.48&17&2.16&2.31\\
\hline
6&2.65&132&2.33&2.45\\
\hline
8&2.75&922&2.38&\newlbd\\
\hline
\end{tabular}
\end{center}

%\subsection{Upper and lower bounds for $\alpha$}
%\label{sec:bounds}
%As we mentioned earlier, we assume that for each type $i$, the marginal unoccupied probability $\alpha_i \in [\frac{1}{1+\lambda},1]$. This is the most trivial range for the possible value of $\alpha$. We can know more about the range of $\alpha$ deduced from this simple fact. By applying the trivial bound to \eqref{eq:unocc}, we know that
%$$\alpha_i = \frac{1}{1 + \lambda \prod_{w\in \bM_i} \alpha_w} \le  \frac{1}{1 + \lambda \left(\frac{1}{1+\lambda}\right)^{\Delta_i} } = B_i,$$
%where $\Delta_i = \sum_j \bM_i(j)$.
%Hence, $\alpha_i$ should be either equal to $1$ or in the interval $[\frac{1}{1+\lambda}, B_i]$. One can also deduce an improved lower bound for the interval by the same logic. This is because when $\alpha_i = \frac{1}{1+\lambda}$, the only possibility is that all of the children of the type $i$ vertex are on the boundary condition and assigned to be unoccupied, i.e., for all $w\in \bM_i$, $\alpha_w = 1$. As we just showed, when $\alpha_w < 1$, indeed $\alpha_w < B_{t(w)}$. Then, the following $L_i$ becomes a natural lower bound for the interval: $$L_i = \frac{1}{1 + \lambda \max_{w\in\bM_i}{ B_{t(w)}}}.$$
%Therefore, we can conclude that there are just three possibilities for the value of $\alpha_i$:
%\begin{itemize}
%\item $\alpha_i = \frac{1}{1+\lambda}$, when all children of type $i$ vertex are unoccupied;
%\item $\alpha_i \in [L_i,B_i]$;
%\item $\alpha_i = 1$, when the vertex is  set to unoccupied or one of the children of it is set to occupied.
%\end{itemize}
\section{Reducing the Size of the LP}
\label{sec:LP-solve}
Initially, when we write down the linear programs (LPs) for the $\bM_8$ matrix with the size of intervals around $10^{-3}$, the number of constraints and variables is huge, approximately $10$ billion constraints and 1 million variables. Solving this LP directly is not possible, as the data will not even fit in memory.
Notice that the LPs we create have high constraint-variable ratios. One standard technique to solve such LPs is to write the dual which has a high variable-constraint ratio and apply the column generation method
\cite{BJSbook}. From the primal point of view, we try to guess the set of tight constraints, by picking a set of primal constraints, solving a smaller LP and checking whether the rest of the constraints are satisfied. When there are violated constraints, several of the most violated constraints are added to the set and we iterate the procedure until the LP is solved.

Using column generation we obtain an LP that can be solved, but running the method takes too long. Next we will present two of our major techniques to reduce the size of the LPs so that we can solve them within a few days.
\subsection{Nonhomogeneous interval size}
In Theorem \ref{thm:LP}, we break the intervals into subintervals of the same size.
The algorithm was designed to start with a very coarse set of the subintervals with a uniform length and if the LP has no solution, then the algorithm will try to decrease the length and re-solve
the new LP. Usually, the algorithm has to make the length as small as $10^{-3}$ for the LP to have a solution. This creates lots of constraints.
Notice that, the constraints are tight only in a very small range of the interval $(\frac{1}{1+\lambda}, 1)$. Therefore, we can try to break the intervals into subintervals of different sizes.

The goal of breaking intervals is to change the primal constraints so that the objective function $v$ can be achieved at a smaller value. In column generation, shadow prices are used for this purpose. However, here deciding which intervals to break, affects the objective in a nonlinear fashion. Thus, we use a heuristic pricing scheme on the intervals to pick which ones to break. The following briefly describes our heuristic approach.

For each interval, we know how many constraints are involved for that interval and how many of them are violated (i.e., $ax - b > 0$). We sum up the values of $ax - b$ for how much each constraint is violated and then scale this by a factor of the size of the interval to obtain what we define as its price. The algorithm will pick several intervals with the highest prices to break. The reason why we scale by a factor which is a function of the size of the interval is that we do not want to break the intervals that are already very small. In Figure \ref{figure-phii}
we show a step function $\Phi_i$ for a type $i$ in $\bD_8$ found by the LPs.  One can observe from the figure that most of the intervals have large lengths; in fact, there are some small intervals in the middle as these are the intervals that create tight constraints.

\begin{figure}[h]
\begin{center}
%\centerline{\psfig{figure=stepfunction.eps,width=10cm,angle=0}}
\includegraphics[width=5in,angle=0]{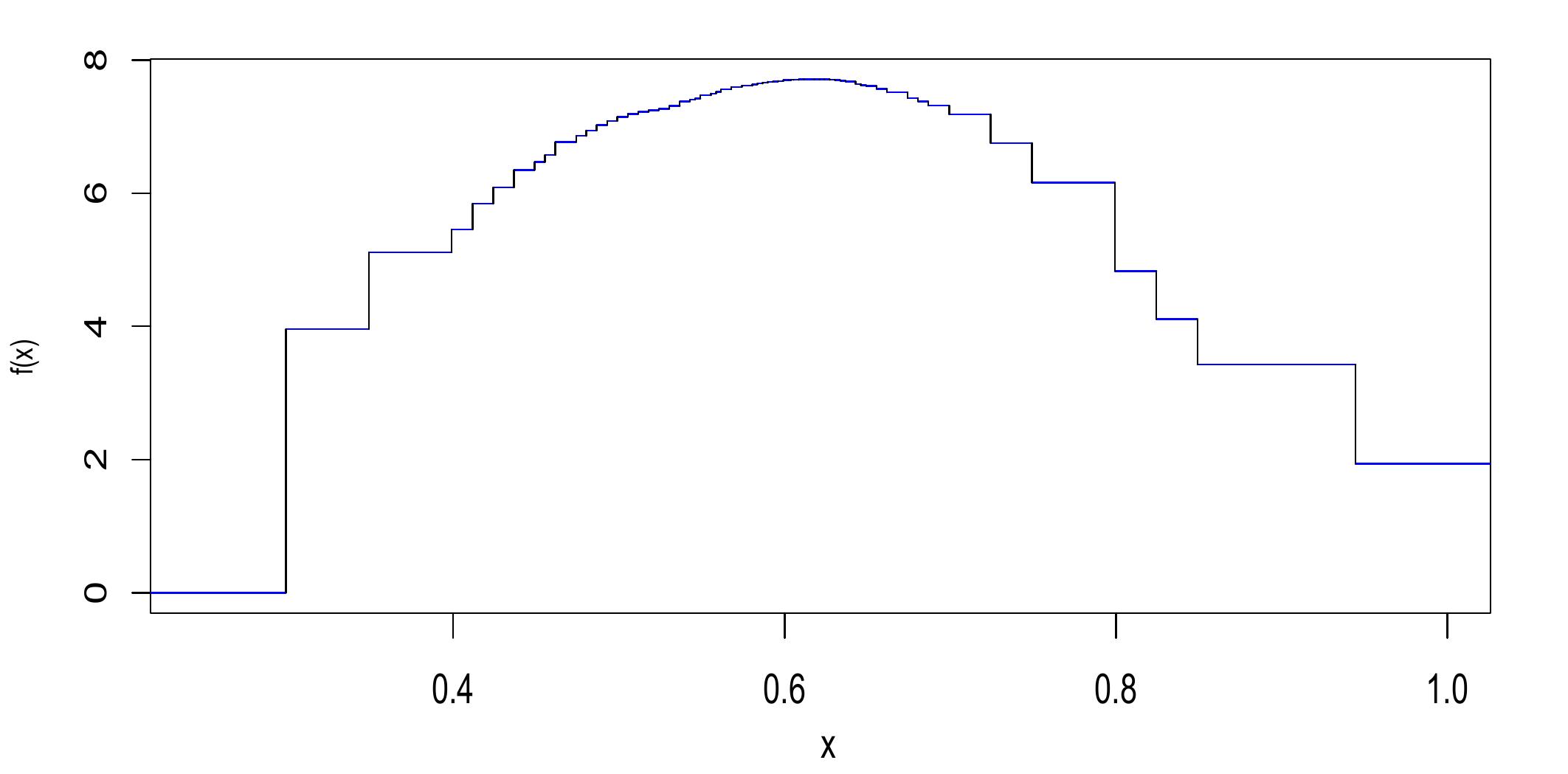}
\caption{A step function $\Phi_i$ found by the LPs}
\label{figure-phii}
\end{center}
\end{figure}

\section{Proof of Contraction Condition Implying SSM}
\label{sec:A}
\begin{proof}[Proof of Lemma \ref{lem:con}]
We fix a tree $T$ in the family $\cF_{\leq\bM}$.  Assuming Condition  \eqref{eq:func} holds, we want to show that
WSM holds for $T$.
Note that Condition \eqref{eq:func} is independent of the tree $T$ we choose.

As in \cite{RSTVY}, we view the boundary condition $\Gamma$ as a continuous parameter.
Hence, throughout the remainder of the proof, $\Gamma\in [0,1]$.
Since we are simply aiming to prove WSM on tree $T$, we can view the boundary as all of
the vertices a fixed distance $L$ from the root of $T$.
Therefore, given a boundary condition $\Gamma$ and for a fixed $L$, we assign the
boundary condition by each vertex at depth $L$ being
fixed to be unoccupied with probability $\Gamma$ and fixed to occupied with probability $1-\Gamma$.
Note that for $L$ even, $\Gamma=1$ corresponds to the even boundary, and
$\Gamma=0$ corresponds to the odd boundary.

Let $\alpha_{i,T}(\Gamma,L)$ be the marginal unoccupied probability for a type $i$ vertex $v$ in the tree
$T_v$ rooted at $v$ where the boundary condition $\Gamma$ is assigned to the vertices at depth $L$ in
$T_v$.
Putting this notation into Equation \eqref{eq:unocc}, for the tree $T$ we have:
\begin{equation}
\label{eq:unoccT}
\alpha_{i,T}(\Gamma,L) = \frac{1}{1 + \lambda \prod_{w \in \bM_i} \alpha_{w,{T}}(\Gamma,L-1)},
\end{equation}
where $\alpha_{w,T}(\Gamma,L)$ equals to $1$ if the vertex $w$ is not in the tree $T$, and otherwise is the marginal unoccupied probability of vertex $w$ in tree $T$ with the fractional boundary condition $\Gamma$.

By integrating over $\Gamma$ we can see that if
\begin{equation}\label{eq:derivcontr}
\left|\frac{\dd \alpha_{i,T}(\Gamma,L)}{d \Gamma}\right| \leq \gamma^L,
\end{equation}
then WSM holds for $T$ at the vertex $v$ of type $i$ since the even and odd boundaries
correspond to $\Gamma=1$ and $\Gamma=0$ (depending on the parity of $L$).

For a vertex $v$ of type $i$, we have the following equation for the derivatives at $\alpha_{i,T}(\Gamma,L)$ with respect to the boundary:
\begin{equation}\label{eq:recDer}
\frac{\dd \alpha_{i,T}(\Gamma,L)}{\dd \Gamma} = -( 1 - \alpha_{i,T}(\Gamma,L))(\alpha_{i,T}(\Gamma,L)) \sum_{w \in \bM_i} \frac{d \alpha_{w, T}(\Gamma,L-1)}{\dd \Gamma} \frac{1}{\alpha_{w,{T}}(\Gamma,L-1)}.
\end{equation}

From \eqref{eq:recDer} it is sufficient to show for all $i$ and all $\alpha_w \in [1/(1+\lambda), 1]$,
\[
( 1 - \alpha_{i})(\alpha_{i}) \sum_{w \in  \bM_i}  \frac{1}{\alpha_{w}} < \gamma
\]
to obtain \eqref{eq:derivcontr} and hence WSM holds for $T$, where in the inequality $\alpha_i = \left(1+\lambda \prod_{w\in \bM_i}\alpha_w\right)^{-1}$. Note that, from here we already obtain a condition that
implies the WSM holds for all trees $T$ in the family $\cF_{\leq \bM}$.

However, technically, it is hard to show the contraction of the above inequality due to the nonhomogeneous marginal distributions $\alpha_w$ from different children vertices as well as the irregular structure of the trees. We instead use a monotonic mapping $\phi_i$ (the messages from a vertex of type $i$ to its parent) for each type $i$, and show that
\begin{equation}\label{eq:derContPhi}
\left|\frac{\dd \phi_i(\alpha_{i,T}(\Gamma,L))}{d \Gamma}\right| \leq \gamma^L,
\end{equation}
which also implies that WSM holds for all trees $T\in \cF_{\le \bM}$.

Setting $\Psi_i(x) = \left(x \cdot  \frac{\dd\phi_i(x)}{ \dd x }\right)^{-1}$,
we have
\[\frac{1}{\alpha_w} = \Psi_{t(w)}(\alpha_w) \frac{\dd \phi_{t(w)}(\alpha_w)}{\dd \alpha_w},
\]
and thus
\[
\frac {\dd \phi_i(\alpha_i)}{\dd \Gamma}= - \frac{( 1 - \alpha_i)}{\Psi_i(\alpha_i)} \sum_{w \in  \bM_i} \Psi_{t(w)}(\alpha_w)\frac{\dd \phi_{t(w)}(\alpha_w)}{\dd \Gamma}.
\]
Notice that to obtain \eqref{eq:derContPhi}, from this last equation we just need
Condition \eqref{eq:func} to be true.
\end{proof}

\section{Conclusions}

Current techniques for proving lower bounds on $\lambda_c(\integers^2)$
analyze SSM on $\Tsaw^\sigma(\integers^2)$.  This paper shows that this approach
will not be sufficient to reach the conjectured threshold of $3.79...$.
One problem in this approach
is that boundary conditions obtainable on $\Tsaw(\integers^2)$
are not necessary realizable on $\integers^2$. Some of the boundary conditions are more ``extremal'' than the one that is on $\integers^2$ which yields a lower weak spatial mixing threshold.
Finding a way to exclude certain boundary conditions for $\Tsaw(\integers^2)$
would be an extremely interesting direction.

\end{document}